\documentclass[twocolumn]{svjour3}          
\smartqed  
\usepackage{graphicx}

\usepackage{amssymb,epsfig,subfigure,amsmath}
\usepackage{color}

\usepackage{amscd}

\def\R{\mathbb R}

\def\N{\mathbb N}

\def\dis{\displaystyle}
\def\al{\alpha}
\def\la{\lambda}
\def\ga{\gamma}
\def\ep{\varepsilon }
\def\th{\theta}
\def\vphi{\varphi}
\def\wt{\widetilde }

\def\de{\delta}

\def\R{\mathbb R}
\def\N{\mathbb N}
\def\I{\mathbb I}

\def\eu{\, \textrm{e}}
\def\span{\hbox{\rm span}}
\def\dis{\displaystyle}
\def\({\left(}
\def\){\right)}

\DeclareMathOperator{\sech}{sech}
\DeclareMathOperator{\sn}{sn}
\DeclareMathOperator{\dn}{dn}
\DeclareMathOperator{\cn}{cn}
\DeclareMathOperator{\Fix}{Fix}
\DeclareMathOperator{\am}{am}

\begin{document}

\title{Dynamics of generalized PT-symmetric dimers with time periodic gain-loss 
}


\author{F.\ Battelli \and J.\ Dibl\' ik \and M.\ Fe\v{c}kan \and J.\ Pickton \and M.\ Posp\' i\v sil \and H.\ Susanto}


\institute{
F.\ Battelli \at Department of Industrial Engineering and Mathematical Sciences, Marche Polytecnic University,
Via Brecce Bianche 1, 60131 Ancona, Italy\\ \email{battelli@dipmat.univpm.it}
\and
J.\ Dibl\' ik \at Department of Mathematics, Faculty of Electrical Engineering and Communication, Brno University of Technology,
Technick\' a 3058/10, 616 00 Brno, Czech Republic\\ \email{diblik@feec.vutbr.cz} \and
M.\ Fe\v{c}kan \at Department of Mathematical Analysis and Numerical Mathematics, Comenius University, Mlynsk\'a dolina, 842 48 Bratislava, Slovakia\\ \email{Michal.Feckan@fmph.uniba.sk} \and
J.\ Pickton \at School of Mathematical Sciences, University of Nottingham, University Park, Nottingham, NG7 2RD, United Kingdom
\and
M.\ Posp\' i\v sil \at Centre for Research and Utilization of Renewable Energy, Faculty of Electrical Engineering and Communication,
Brno University of Technology, Technick\' a 3058/10, 616 00 Brno, Czech Republic\\ \email{pospisilm@feec.vutbr.cz}
\and
H.\ Susanto \at Department of Mathematical Sciences, University of Essex, Wivenhoe Park,
Colchester, CO4 3SQ, United Kingdom\\ \email{hsusanto@essex.ac.uk}}
\date{Received: date / Accepted: date}

\maketitle

\begin{abstract}

A Parity-Time (PT)-symmetric system with periodically varying-in-time gain and loss modeled by two coupled Schr\"odinger equations (dimer) is studied. It is shown that the problem can be reduced to a perturbed pendulum-like equation. This is done by finding two constants of motion. Firstly, a generalized problem using Melnikov type analysis and topological degree arguments is studied for showing the existence of periodic (libration), shift periodic (rotation), and chaotic solutions. Then these general results are applied to the PT-symmetric dimer. It is interestingly shown that if a sufficient condition is satisfied, then rotation modes, which do not exist in the dimer with constant gain-loss, will persist. An approximate threshold for PT-broken phase corresponding to the disappearance of bounded solutions is also presented. Numerical study is presented accompanying the analytical results.

\keywords{PT-symmetry \and PT-reversibility \and Schr\"odinger equation \and Melnikov function \and perturbation \and chaos}
\end{abstract}

\section{Introduction}\label{intr}

We study the coupled Schr\"odinger equations
\begin{equation}\label{ie1}
\begin{array}{l}
\dis\imath\dot u_1=-u_2-|u_1|^2u_1-\imath\ga(t)u_1\\
\dis\imath\dot u_2=-u_1-|u_2|^2u_2+\imath\ga(t)u_2,
\end{array}
\end{equation}
with a Kerr non-linearity where $\ga(t)$ is a gain-loss parameter. It is easy to see that \eqref{ie1} is invariant under the exchange
\begin{equation}\label{ie1b}
(u_1, u_2, \dot{u}_1, \dot{u}_2, t, \imath) \mapsto (u_2, u_1,  -\dot{u}_2, -\dot{u}_1, -t, -\imath),
\end{equation}
known as Parity-Time (PT) symmetry \cite{bend07}, provided that
$\gamma(t)$ is even. However, for the sake of completeness, in
this paper we consider a general, possibly not even, $\ga(t)$. In
this case, we may consider a formal symmetry in the sense that
under the transformation
$$
(u_1, u_2, \dot{u}_1, \dot{u}_2, t, \imath) \mapsto (u_2, u_1,  \dot{u}_2, \dot{u}_1, -t, -\imath),
$$
the new equations have $(u_1(-t),u_2(-t))$ as a solution whenever
$(u_1(t),u_2(t))$ solves \eqref{ie1}. In this context, the
PT-symmetry \cite{bend07} can be referred to as PT-reversibility.

Equations \eqref{ie1} without nonlinearity with time-in\-de\-pen\-dent $\ga$ were derived in \cite{gana07} as a finite dimensional reduction of mode guiding optical systems with gain and loss regimes \cite{rusc05,klai08}. The following successful experiments \cite{guo09,rute10} have stimulated extensive studies on the PT systems, including the interplay between the symmetry and nonlinearity \cite{li11,duan13,rame10,sukh10,miro11,alex14}. Later, equations similar to \eqref{ie1} were also derived \cite{rodr12} modeling Bose-Einstein condensates with PT-symmetric double-well potentials \cite{wunn12,grae12,heis13,dast13}.

The study of the above systems with equal and constant gain and loss were motivated by a relatively recent postulate in quantum physics that Hamiltonian operators do not have to be self-adjoint (Hermitian) for the potential physical relevance; non-Hermitian operators only need to respect Parity (P) and time-reversal (T) symmetries and the spectrum of the operators can be real if the strength of the anti-Hermitian part, i.e.\ $\gamma\equiv const$ in \eqref{ie1}, does not exceed a critical value \cite{bend98,bend99,bend07}. Above the critical value, the systems are said to be in the PT-broken phase regime.

Note that system \eqref{ie1} with constant $\gamma$ was firstly proposed and studied in \cite{chen92,jorg93,jorg94}. Equation \eqref{ie1} was a particular case of a notably integrable dimer, where it was shown that the general system could be reduced to a first-order differential equation with polynomial nonlinearity and possesses blow-up solutions \cite{jorg93,jorg94}. The solution dynamics of \eqref{ie1} in that case has been studied in \cite{P,kevr13,bara14,bara13}.

Recently, \eqref{ie1} with a gain and loss parameter that is periodically a function of time was considered in \cite{horn13}. By deriving an averaged equation, the rapid modulation in time of the gain and loss profile was shown to provide a controllable expansion of the region of exact PT-symmetry, depending on the strength and frequency of the imposed modulation. In the presence of dispersion along the direction perpendicular to the propagation axis, rapidly modulated gain and loss was shown to stabilize bright solitons \cite{drib11}. The same system featuring a "supersymmetry" when the gain and loss are equal to the inter-core coupling admits a variety of exact solutions (we focus on solitons), which are subject to a specific subexponential instability \cite{drib112}. However, the application of the "management" in the form of periodic simultaneous switch of the sign of the gain, loss, and inter-core coupling, effectively stabilizes solitons, without destroying the supersymmetry \cite{drib112}. A significant broadening of the quasienergy spectrum, leading to a hyperballistic transport regime, due to a time-periodic modulation has also been reported in \cite{vall13}. When the damping and gain factor varies stochastically in time, it was shown in \cite{kono14} that the statistically averaged intensity of the field grows. A different, but similar dimer to \eqref{ie1} involving a parametric (ac) forcing through periodic modulation of the coupling coefficient of the two degrees of freedom was considered recently in \cite{damb13}. It was reported that the combination of the forcing and PT-symmetric potential exhibits dynamical charts featuring tongues of parametric instability, whose shape has some significant differences from that
predicted by the classical theory of the parametric resonance in conservative systems.

In this work, we consider \eqref{ie1} with time-varying $\ga$ whose modulation period is $\mathcal{O}(1)$. Using Melnikov type analysis and topological degree arguments, we prove and derive conditions for the existence of periodic (libration), shift periodic (rotation), and chaotic solutions in the system. Note that the existence of bounded rotation modes for the phase-difference of the dimer wavefunctions is a novel and interesting feature as they cease to exist in the time-independent $\ga$ \cite{P}. We also derive an approximate threshold for PT-broken phase corresponding to the disappearance of bounded solutions. 

The paper is organized as follows. In Section \ref{red}, we reduce the governing equation \eqref{ie1} into
\begin{equation}\label{ee1ee}
\begin{array}{l}
\dis \ddot \psi=4\(c^2\sin\psi+\ga(t)z\),\\
\dis \dot z=\ga(t)\dot\psi,
\end{array}
\end{equation}
which becomes the main object which we study in this paper. In Section \ref{bif}, we consider a more general problem than \eqref{ee1ee} of the form
\begin{equation}\label{ie1e}
\begin{array}{l}
\dis\ddot\psi = f(\psi) +\ga(t)h(z) \\
\dis \dot z=\ga(t)\dot\psi
\end{array}
\end{equation}
where $\ga(t)$ is $1-$periodic and $f(\psi)$, $h(z)$ are specified later. First we consider small forcings (see \eqref{e1}) and we derive Melnikov type conditions for bifurcation of subharmonics, shift-subharmonics and heteroclincs/homo\-clinics. We also discuss the persistence of these solutions when \eqref{ie1e} is reversible symmetric. Then an averaging-like result concerning the existence of periodic solutions is proved when all terms of \eqref{ie1e} are small (see \eqref{e13}). We return back to general \eqref{ie1e} in Section \ref{for} and use topological degree arguments for showing odd-periodic, shift-periodic and anti-periodic solutions under certain symmetric properties of \eqref{ie1e}. In Section \ref{appl}, we apply results of Sections \ref{bif} and \ref{for} to \eqref{ee1ee}. In Section \ref{numer} we present numerical integrations of \eqref{ee1ee} illustrating the analytical results derived in the previous sections.

\section{Model reduction}\label{red}

To simplify \eqref{ie1}, by following \cite{P}, we use polar form
$$
u_j(t)=|u_j(t)|\eu^{\imath\phi_j(t)},\quad j=1,2,
$$
by assuming $u_1(t)u_2(t)\ne0$. Clearly
$$
\dot u_j(t)=\(\frac{d}{dt}|u_j(t)|+\imath|u_j(t)|\frac{d\phi_j(t)}{dt}\)\eu^{\imath\phi_j},\quad j=1,2.
$$
Hence setting $\th(t):=\phi_2(t)-\phi_1(t)$ and separating \eqref{ie1} into the real and imaginary parts, we have
\begin{equation}\label{ie2}
\begin{array}{l}
\dis \frac{d}{dt}|u_1(t)|=-|u_2(t)|\sin\th(t)-\ga(t)|u_1(t)|,\\
\dis \frac{d}{dt}|u_2(t)|=|u_1(t)|\sin\th(t)+\ga(t)|u_2(t)|,\\
 \dis\frac{d}{dt}\th(t)=\(|u_2(t)|^2-|u_1(t)|^2\)\(1-\frac{\cos\th(t)}{|u_1(t)||u_2(t)|}\).
\end{array}
\end{equation}

Then it can be shown
$$
\frac{d}{dt}\(|u_1(t)|^2|u_2(t)|^2-2|u_1(t)||u_2(t)|\cos\th(t)\)=0.
$$
So we obtain
\begin{equation}\label{ie3}
c^4=1+|u_1(t)|^2|u_2(t)|^2-2|u_1(t)||u_2(t)|\cos\th(t).
\end{equation}
Note
$$
\begin{array}{l}
\dis 1+|u_1(t)|^2|u_2(t)|^2-2|u_1(t)||u_2(t)|\cos\th(t)\\
\dis =\(1-|u_1(t)||u_2(t)|\)^2+2(1-\cos\th(t))|u_1(t)||u_2(t)|\ge0.
\end{array}
$$
Hence the constant $c\ge0$ in \eqref{ie3} is well defined. We study $c>0$ in this paper. Note $c\ne1$ implies $u_1(t)u_2(t)\ne0$. Next, using \eqref {ie3}, we introduce a new variable $\psi$ as follows
\begin{equation}\label{ie4}
\begin{array}{l}
\dis c^2\sin\psi(t)=|u_1(t)||u_2(t)|\sin\th(t),\\
\dis c^2\cos\psi(t)=|u_1(t)||u_2(t)|\cos\th(t)-1.
\end{array}
\end{equation}
From \eqref{ie2} and \eqref{ie4}, we obtain
$$
\begin{array}{ll}
\dis \(|u_1(t)|^2-|u_2(t)|^2\)|u_1(t)||u_2(t)|\sin\th(t)\\
\dis =\frac{d}{dt}\(|u_1(t)||u_2(t)|\cos\th(t)\) =\frac{d}{dt}\(c^2\cos\psi(t)\)\\
\dis =-c^2\dot\psi(t)\sin\psi(t)=-\dot\psi(t)|u_1(t)||u_2(t)|\sin\th(t),
\end{array}
$$
which implies
\begin{equation}\label{ie5}
\dot\psi(t)=|u_2(t)|^2-|u_1(t)|^2.
\end{equation}
Differentiating \eqref{ie5} and using \eqref{ie2}, \eqref{ie4}, we arrive at
\begin{equation}\label{ie6}
\begin{array}{l}
\dis \ddot\psi(t)=4|u_1(t)||u_2(t)|\sin\th(t)+2\ga(t)\(|u_1(t)|^2+|u_2(t)|^2\)\\
\dis =4c^2\sin\psi(t)+2\ga(t)\(|u_1(t)|^2+|u_2(t)|^2\).
\end{array}
\end{equation}
On the other hand, using \eqref{ie2} and \eqref{ie5}, we obtain
\begin{equation}\label{ie7}
\begin{array}{l}
\dis \frac{d}{dt}\(|u_1(t)|^2+|u_2(t)|^2\)\\
\dis =2\ga(t)\(|u_2(t)|^2-|u_1(t)|^2\)=2\ga(t)\dot\psi(t).
\end{array}
\end{equation}
Finally setting
$$
z(t):=\frac{|u_1(t)|^2+|u_2(t)|^2}{2},
$$
\eqref{ie6} and \eqref{ie7} give \eqref{ee1ee}.

\begin{remark} Introducing the power $P:=|u_1|^2+|u_2|^2$ and the probability difference between the two waveguides $\Delta:=|u_2|^2-|u_1|^2$, and using \eqref{ie2}, \eqref{ie5}-\eqref{ie7}, instead of \eqref{ee1ee}, we obtain from \eqref{ie1} \cite{P}
\begin{equation}\label{ie8}
\begin{array}{l}
\dis \dot P=2\ga(t)\Delta\\
\dis \dot \Delta=2\ga(t)P+2\sqrt{P^2-\Delta^2}\sin\th\\
\dis \dot \th=\Delta\(1-\frac{2\cos\th}{\sqrt{P^2-\Delta^2}}\),
\end{array}
\end{equation}
which reduces to the common pendulum-like system for wavefunctions in a double-well potential when $\gamma\equiv0$ (see \cite{P} and references therein). Note $P^2-\Delta^2=4|u_1|^2|u_2|^2$, so $P^2>\Delta^2$ if and only $|u_1||u_2|\ne0$. Then
$$
c^4=1+\frac{P^2(t)-\Delta^2(t)}{4}-\sqrt{P^2(t)-\Delta^2(t)}\cos\th(t).
$$
Furthermore, one can compare the first equation of \eqref{ie8} with \eqref{ie5} to obtain
\begin{eqnarray}
\dot{P}=\frac{d}{dt}\left(2\gamma\psi\right)-2\dot{\gamma}\psi,
\label{tamb1}
\end{eqnarray}
which can be integrated to yield
\begin{eqnarray}
P = 2\gamma\psi - 2\int \dot{\gamma}\psi\,dt.\
\label{eq:Ppsi}
\end{eqnarray}
Using equations \eqref{ie4}, \eqref{ie5} and \eqref{eq:Ppsi} in the second equation of \eqref{ie8} will yield the equation
\begin{eqnarray}
\frac{1}{4}\ddot{\psi} = c^2\sin\psi+\gamma ^2\psi - \gamma \int \dot{\gamma}\psi\,dt,
\end{eqnarray}
which becomes \eqref{ee1ee} with $z(t)=\gamma\psi - \int \dot{\gamma}\psi\,dt$.
\end{remark}

\section{General bifurcation results and related topics}\label{bif}

We study the coupled system
\begin{equation}\label{e1}
\begin{array}{l}
\dis\ddot\psi = f(\psi) + \ep\ga(t)h(z) \\
\dis \dot z=\ep\ga(t)\dot\psi
\end{array}
\end{equation}
where $\ga(t)$ is $1-$periodic and $C^2$-smooth, $f(\psi)$, $h(z)$ are $C^2$ and $\ep\in\R$ is small. Set:
\[
F(\psi)=\int_0^\psi f(u) du, \quad H(z) = \int_0^z h(u)du.
\]
From \eqref{e1} it easy to derive the relation
$$
\frac{d}{dt}\left[\dot\psi^2-2F(\psi)-2H(z)\right]=0,
$$
i.e.,
\begin{equation}\label{tamb3}
\dot\psi^2-2F(\psi)-2H(z) = 2d
\end{equation}
for a constant $d\in\R$. Of course, \eqref{tamb3} holds even if $\ga(t)$ is not periodic.
\begin{remark} The same first integral is such for the system
\begin{equation}\label{e1bb}
\begin{array}{l}
\ddot\psi = f(\psi) +\ep\ga(t)h(z)h_1(z)    \\
\dot z = \ep\ga(t)h_1(z)\dot\psi
\end{array}
\end{equation}
where $h_1(z)$ is $C^2$.
\end{remark}

We assume that
\begin{itemize}
\item[$(C_1)$] $H:I\to J$ is a diffeomorphism from the open interval $I\subseteq\R$ into the open interval $J\subseteq\R$, i.e., $h(z)\ne 0$ for any $z\in I$.
\end{itemize}
Then we have:
\begin{equation}\label{e1b}
z=H^{-1}\(\frac{\dot\psi^2}{2}-F(\psi)-d\),
\end{equation}
provided $\frac{\dot\psi^2}{2}-F(\psi)-d\in J$ and \eqref{e1} is reduced to
\begin{equation}\label{e1c}
\ddot\psi = f(\psi) + \ep\ga(t) h\(H^{-1}\(\frac{\dot\psi^2}{2}-F(\psi)-d\)\) .
\end{equation}
This equation must be solved so that its solution satisfies $\frac{\dot\psi^2(t)}{2}-F(\psi(t))-d\in J$ for any $t$ in some interval (or $t\in\R$ if we look for globally defined solutions).

Setting $\ep=0$ into \eqref{e1c} we obtain the unperturbed equation:
\begin{equation}\label{e2}
\ddot \psi=f(\psi).
\end{equation}
We assume that
\begin{itemize}
\item[($H_1$)] equation \eqref{e2} has a family of periodic solutions
\begin{equation}\label{e3}
\psi_\la(t)=\psi_\la(t+T_\la)
\end{equation}
with periods $T_\la\in\R$ and $\dot\psi_\la(t)$ is the unique, up to a multiplicative constant, $T_\la-$periodic solution of the variational equation $\ddot\psi = f'(\psi_\la(t))\psi$.
\end{itemize}
\begin{remark} Assuming that dependence on $\la$ in \eqref{e3} is $C^1$-smooth, from
\eqref{e2} we derive that $\dot\psi_\la(t)$ and $\frac{\partial}{\partial\la}\psi_\la(t)$ solve  $\ddot\psi = f'(\psi_\la(t))\psi$. But from \eqref{e3} we have
$$
\begin{gathered}
\dot\psi_\la(t)=\dot\psi_\la(t+T_\la),\\
\frac{\partial}{\partial\la}\psi_\la(t)=\frac{\partial}{\partial\la}\psi_\la(t+T_\la)+\dot \psi_\la(t+T_\la)\frac{\partial}{\partial\la}T_\la.
\end{gathered}
$$
Hence $\frac{\partial}{\partial\la}\psi_\la(t)$ is $T_\la$-periodic if and only if $\frac{\partial}{\partial\la}T_\la=0$. This observation is well-known but we presented it here for the reader's convenience.
\end{remark}
Next we assume that for some $\la\in\R$ the resonance condition
\begin{itemize}
\item[($R$)] \quad $T_\la = \frac{p}{q}$ for $p,q\in\N$
\end{itemize}
holds. Using \eqref{e2} we see that $c_\la\in\R$ exists so that
\[
\frac{\dot\psi^2_\la(t)}{2}-F(\psi_\la(t))=c_\la .
\]
Then we assume:
\begin{itemize}
\item[($H_2$)] $c_\la-d\in J$.
\end{itemize}
Condition ($H_2$) guarantees that
$$
\frac{\dot\psi_\la^2(t)}{2}-F(\psi_\la(t))-d\in J
$$
for any $t\in\R$.

Now, we follow the standard subharmonic Melnikov method \cite{Ber,Ch,GH} to \eqref{e1c} based on the Lyapunov-Schmidt method, but for reader's convenience we present more details. First, we take Banach spaces
$$
\begin{array}{c}
\dis X:=\left\{\psi\in C^2(\R,\R)\mid \psi(t+T)=\psi(t)\, \forall t\in\R\right\}\\
\dis Y:=\left\{\psi\in C(\R,\R)\mid \psi(t+T)=\psi(t)\, \forall t\in\R\right\}
\end{array}
$$
with the usual maximum norms $\|\psi\|_2$ and $\|\psi\|_0$, respectively, where $T=p=qT_\la$, so that $\psi_\la(t)$ is $T-$periodic. We want to find solutions in $X$ of \eqref{e1c} such that
\begin{equation}\label{e4a}
\frac{\dot\psi^2(t)}{2}-F(\psi(t))-d \in J
\end{equation}
and $\psi(t+T)=\psi(t)$, for any $t\in\R$. We consider the open subset $\Omega\subset X$ of those $\psi\in X$ such that \eqref{e4a} holds. For any $\psi\in\Omega$ we set
\[
\begin{array}{l}
\dis {\cal F}(\psi,\ep)(t) \\
\dis = \ddot\psi - f(\psi) - \ep\ga(t) h\(H^{-1}\(\frac{\dot\psi^2}{2}-F(\psi)-d\)\).
\end{array}
\]
Note that, for any $\al\in\R$ the function $\psi_\la(\cdot +\al):t\mapsto \psi_\la(t+\al)$ satisfies \eqref{e4a} and
\[
{\cal F}(\psi_\la(\cdot + \al),0)(t) = 0 .
\]
Then ${\cal F}(\psi,0)$ vanishes on the one-dimensional manifold ${\cal M} = \{ \psi_\la(\cdot +\al) \mid \al\in\R\}$. The tangent space to ${\cal M}$ at the point  $\psi_\la(\cdot+\al_0)$ is
\[
T_{\psi_\la(\cdot +\al_0)}{\cal M} = \span\{\dot\psi_\la(\cdot +\al_0)\}
\]
and of course the linear map
\[
L:\psi \mapsto \ddot\psi - f'(\psi_\la(t+\al_0))\psi
\]
vanishes on $T_{\psi_\la(\cdot +\al_0)}{\cal M}$. From ($H_1$) it follows that
\[
{\cal N}L=T_{\psi_\la(\cdot +\al_0)}{\cal M}.
\]
Next, it is well known (see \cite{Ch,GH}) that $b(t)\in Y$ belongs to ${\cal R}L$ if and only if
\[
\dis \int_0^T u(t)b(t)\,dt = 0
\]
for any $T-$periodic solution $u(t)$ of $\ddot \psi = f'(\psi_\la(t+\al_0))\psi$. But then, since the space of periodic solutions of this equation is spanned by $\dot\psi_\la(\cdot +\al_0)$, we conclude that the equation
\[
\ddot\psi-f'(\psi_\la(t+\al_0))\psi = b(t)
\]
has a periodic solution (i.e. $b(t)\in{\cal R}L$) if and only if
\[
\int_0^T \dot\psi_\la(t+\al_0)b(t) dt = 0.
\]
Finally,
\[\begin{array}{l}
\dis \frac{\partial}{\partial\ep}{\cal F}(\psi_\la(\cdot+\al_0),0)\\ \dis = -\ga(t) h\(H^{-1}\(\frac{\dot\psi_\la^2(\cdot+\al_0)}{2}-F(\psi_\la(\cdot+\al_0))-d\)\)         \\
\dis = -\ga(t) h(H^{-1}(c_\la-d)).
\end{array}\]
Then, from the well-known Poincar\'e-Melnikov theory we obtain the following result:
\begin{theorem}\label{th1} Suppose that $(H_1)$, $(R)$ and $(H_2)$ hold and that the Melnikov function
\[
M_\la(\al) := \int_0^T \dot\psi_\la(t +\al)\ga(t)dt
\]
has a simple zero at some $\al=\al_0$. Then there exists $\rho>0$ such that for any $\ep$ sufficiently small there exists a unique $\al(\ep)$ ($C^1$ with respect to $\ep$) such that $|\al(\ep)-\al_0|\to 0$ as $\ep\to 0$ and such that equation \eqref{e1c} has a unique $p-$periodic solution $\psi(t,\ep)$ satisfying:
\[
\|\psi(t,\ep)-\psi_\la(t+\al(\ep))\|_2 \le \rho.
\]
Moreover $\|\psi(t,\ep)-\psi_\la(t+\al(\ep))\|_2 \to 0$ as $\ep\to 0$.
\end{theorem}

Now suppose
\begin{itemize}
\item[$(C_2)$] $f(\psi+2\pi)=f(\psi)$, $F(\psi+2\pi)=F(\psi)$,
\item[$(H_3)$] equation \eqref{e2} has a family of $T_\la^l-$shift-periodic solutions
\begin{equation}\label{e3bb}
\psi_\la^l(t+T_\la^l)=\psi_\la^l(t)+2\pi
\end{equation}
with shift-periods $T_\la^l\in\R$ and $\dot\psi_\la^l(t)$ is the unique, up to a multiplicative constant, $T_\la^l-$periodic solution of the variational equation $\ddot\psi = f'(\psi_\la^l(t))\psi$.
\end{itemize}
Using again \eqref{e2} we see that $c_\la^l\in\R$ exists so that
\[
\frac{\dot\psi^{l2}_\la(t)}{2}-F(\psi_\la^l(t))=c_\la^l.
\]
Then Theorem \ref{th1} can be directly extended as follows.

\begin{theorem}\label{th1b} Suppose that $(C_2)$, $(H_3)$ and
\begin{itemize}
\item[$(R')$] \quad $T_\la^l = \frac{p}{q}$ for $p,q\in\N$,
\item[$(H_2')$] $c_\la^l-d\in J$
\end{itemize}
hold and that the Melnikov function
\[
M_\la^l(\al) := \int_0^T \dot\psi_\la^l(t +\al)\ga(t)dt
\]
has a simple zero at some $\al=\al_0$. Then there exists $\rho>0$ such that for any $\ep$ sufficiently small there exists a unique $\al^l(\ep)$ ($C^1$ with respect to $\ep$) such that $|\al^l(\ep)-\al_0|\to 0$ as $\ep\to 0$ and such that equation \eqref{e1c} has a unique $p-$shift-periodic solution $\psi(t,\ep)$ satisfying:
\[
\|\psi(t,\ep)-\psi_\la^l(t+\al^l(\ep))\|_2 \le \rho.
\]
Moreover $\|\psi(t,\ep)-\psi_\la^l(t+\al^l(\ep))\|_2 \to 0$ as $\ep\to 0$.
\end{theorem}

Similarly we can study heteroclinic Melnikov bifurcation \cite{BL1,Ch,G,GH} to \eqref{e1}. In this case we assume

\begin{itemize}
\item[$(H_4)$]  the unperturbed equation \eqref{e2} has a heteroclinic orbit $\psi_{\infty}(t)$ such that $\dis\lim_{t\to\pm\infty}\psi_0(t) = \psi_\pm$ and $\psi_\pm$ are hyperbolic saddles of $\ddot\psi=f(\psi)$.
\end{itemize}

Note that if $\psi_+=\psi_-$ then $\psi_{\infty}(t)$ is a homoclinic orbit. Of course the quantity
\[
\frac{\dot\psi^2}{2} -F(\psi)
\]
is constant along $\psi_\infty(t)$. Setting
\[
c_\infty = \frac{\dot\psi_\infty(t)^2}{2} -F(\psi_\infty(t))
\]
we assume that

\begin{itemize}
\item[$(H_5)$] $c_\infty-d\in J$.
\end{itemize}

Now we take functional spaces $X=C_b^2(\R,\R)$ and $Y=C_b(\R,\R)$ and repeating the above arguments, we have the following result.
\begin{theorem}\label{th2}
Under  assumptions $(H_4)$ and $(H_5)$, there exists $\rho>0$ such that equation \eqref{e1c} has unique hyperbolic periodic solutions $u_\pm(t)$ in a neighborhood of $\psi_\pm$ of radius $\rho$. Moreover, if the Melnikov function
$$
M_\infty(\al):=\int_{-\infty}^\infty \dot\psi_\infty(t+\al)\ga(t)dt
$$
has a simple zero at some $\al=\al_0$ then for any $\ep$ sufficiently small there exists a unique $\al(\ep)$ ($C^1$ with respect to $\ep$) such that $|\al(\ep)-\al_0|\to 0$ as $\ep\to 0$ and such that equation \eqref{e1c} has a unique heteroclinic (homoclinic) solution $\psi(t,\ep)$ satisfying:
\[
\|\psi(t,\ep)-\psi_\infty(t+\al(\ep))\|_2 \le \rho.
\]
Moreover $\|\psi(t,\ep)-\psi_\infty(t+\al(\ep))\|_2 \to 0$ as $\ep\to 0$ and
\[
\lim_{t\to\pm\infty} |\psi(t,\ep)-u_\pm(t)| = 0.
\]
\end{theorem}
\begin{remark}\label{rr1}
As a matter of fact, we have $1$-parametric families of subharmonic, shift-subharmonics and bounded solutions parameterized by $d$. Moreover, by \cite[p. 197]{GH}, these bounded solutions are accumulated with families of subharmonics, shift-subharmonics from Theorem \ref{th1}. Furthermore, Theorem \ref{th2} implies the existence of chaos near these bounded solutions creating heteroclinic cycles (see \cite{BF1} for more details).
Now, $\ga(t)$ can be almost-periodic.
\end{remark}

\begin{remark}\label{rr2}
Note that
\begin{eqnarray*}
&&\int_0^TM_\la(\al)d\al=\int_0^T\int_0^T\dot{\psi}_\la(t+\al)\ga(t)dtd\al\\&=&\int_0^T\ga(t)\(\int_0^T\dot{\psi}_\la(t+\al)d\al\)dt\\
&=&\int_0^T\ga(t)\(\psi_\la(t+T)-\psi_\la(t)\)dt\\&=&\int_0^{qT_\la}\ga(t)\(\psi_\la(t+qT_\la)-\psi_\la(t)\)dt=0,\\
&&\int_0^TM_\la^l(\al)d\al=\int_0^{qT_\la^l}\ga(t)\(\psi_\la^l(t+qT_\la^l)-\psi_\la(t)\)dt\\&=&2\pi q\int_0^{T}\ga(t)dt,\\
&&\int_0^1M_\infty(\al)d\al=\int_{-\infty}^{\infty}\ga(t)\(\psi_\infty(t+1)-\psi_\infty(t)\)dt\\
&=&\int_{-\infty}^{\infty}\ga(t-1)\psi_\infty(t)dt-\int_{-\infty}^\infty\ga(t)\psi_\infty(t)dt=0.
\end{eqnarray*}
So if either the subharmonic Melnikov function $M_\la(\al)$ is not identically zero or $\int_0^{T}\ga(t)dt=0$ and the shift-subharmonic Melnikov function $M_\la^l(\al)$ is not identically zero then they change signs over $[0,T]$ and hence we have existence (but not uniqueness) results similar to Theorems \ref{th1} and \ref{th1b}. The same argument holds for the heteroclinic Melnikov function $M_\infty(\al)$ from Theorem \ref{th2} with the corresponding chaotic behavior (see \cite{BF1} for more details).
\end{remark}

\begin{remark} \label{rr3}
Writing \eqref{e1c} as a system
\begin{equation}\label{e11}
\dot\xi=f(\xi,t,\ep),\quad \xi=(\psi,\th),
\end{equation}
for
\[
f(\xi,t,\ep) := \begin{pmatrix} \theta \\ f(\psi) +\ep\ga(t) h\(H^{-1}\(\frac{\theta^2}{2}-F(\psi)-d\) \)
\end{pmatrix}
\]
we see that, when $\ga(-t)=\ga(t)$, i.e., when $\ga(t)$ is even,
\begin{equation}\label{e12}
R_1f(\xi,t,\ep)=-f(R_1\xi,-t,\ep)
\end{equation}
for the reflection $R_1(\psi,\th)=(\psi,-\th)$. Condition \eqref{e12} means, that \eqref{e11} is $R_1$-reversible \cite{De,F2,F3,WF}. We may suppose that the periodic solutions in \eqref{e3} satisfy $\dot\psi_\la(0)=0$. Then setting $\xi_\la(t)=(\psi_\la(t),\dot\psi_\la(t))$, all $\xi_\la(t)$ intersect transversally the fixed point set $\Fix R_1=\{\xi\in\R^2\mid R_1\xi=\xi\}=\{(\vphi,0)\}$ in 2 points.

So each of these periodics is $R_1$-symmetric, i.e., they satisfy $\xi_\la(t)=R_1\xi_\la(-t)$ since $\xi_\la(0)\in\Fix R_1$. So they persist in \eqref{e11} for $\ep$ small when $\ga(t)$ is even. Furthermore, if $\psi_\infty(t)$ is homoclinic then we may again suppose that $\dot\psi_\infty(0)=0$. Hence $\xi_\infty(t):=(\psi_\infty(t),\dot\psi_\infty(t))$ intersects transversally the fixed point set $\Fix R_1$ in one point. So, it is $R_1$-symmetric and it persists in \eqref{e11} for $\ep$ small. Moreover there is an accumulation of these $R_1$-symmetric periodics on the $R_1$-symmetric homoclinic.

Next supposing that $f(\psi)$ and $\ga(t)$ are odd, i.e., $f(\psi)=-f(-\psi)$ and $\ga(t)=-\ga(-t)$, and taking the reflection $R_2(\vphi,\th)=(-\vphi,\th)$, we see that \eqref{e11} is $R_2$-reversible. Note $F(-\psi)=F(\psi)$, i.e., $F(\psi)$ is even. Then if the periodic solutions \eqref{e3} and heteroclinic cycle $\{\xi_\infty(t)\mid t\in\R\}\cup\{R_2\xi_\infty(t)\mid t\in\R\}$ intersect transversally the fixed point set $\Fix R_2=\{(0,\th)\}$,  that is, if each of these orbits is $R_2$-symmetric, they persist in \eqref{e11} for $\ep$ small when $\ga(t)$ is odd. Of course, here we suppose that $\xi_\la(0),\xi_\infty(0)\in\Fix R_2$. Moreover, there is again an accumulation of $R_2$-symmetric periodic orbits on the $R_2$-symmetric heteroclinic cycle. Finally, if in addition $(C_2)$ and $(H_3)$ hold then we have the same persistence and accumulation results for shift-subharmonics.
\end{remark}

\medskip
Finally, instead of \eqref{e1} we consider the equation
\begin{equation}\label{e13}
\begin{array}{l}
\dis \ddot\psi = \ep f(\psi,\ep)+\ep\ga(t)h(z)  \\
\dis \dot z=\ep\ga(t)\dot\psi
\end{array}
\end{equation}
where $\ga(t)$ is $1-$periodic and $C^2$-smooth, $\ep>0$ is small and $f(\psi,\ep)$ is again $C^2$-smooth. Arguing as before, \eqref{e13} is reduced to (see \eqref{e1c})
\begin{equation}\label{e14}
\ddot\psi=\ep f(\psi,\ep) + \ep\ga(t)h\(H^{-1}\(\frac{\dot\psi^2}{2}-\ep F(\psi)-d\)\).
\end{equation}
Note that, for $\ep=0$ equation \eqref{e14} has the $1-$periodic solutions $\psi(t)=\al$. Suppose that \eqref{e14} has a $1-$periodic solution for any $\ep>0$ small uniformly bounded by $M$, so $\|\psi\|_2\le M$. Since $\int_0^1\dot\psi(t)dt=0$, there is an $t_0\in[0,1)$ such that $\dot\psi(t_0)=0$. Then we derive
\begin{equation}\label{e15a}
\begin{array}{l}
\dis\dot\psi(t)= \ep \int_{t_0}^t f(\psi(s),\ep)\\
\dis+ \ga(s) h \(H^{-1}\(\frac{\dot\psi(s)^2}{2}-\ep F(\psi(s))-d\)\) ds
\end{array}
\end{equation}
and hence taking absolute values and using the boundedness of $\psi(t)$, and the smoothness of $f$, $h$ and $H^{-1}$ we see that
\[
\|\dot\psi\|_0=O(\ep)
\]
but then $-d\in \bar J$. We suppose (see $(H_2)$)
\begin{equation}\label{c3}
-d\in J.
\end{equation}
Next we note that the equation
\[
\ddot\psi=b(t)=b(t+1)
\]
has a $1-$periodic solution if and only if $\int_0^1 b(t)dt=0$. So, repeating the above method for \eqref{e15a} with the splitting
$$
\begin{array}{c}
\dis Y=Y_1\oplus Y_2,\quad X=X_1\oplus X_2,\\
\dis Y_1:=\span\{1\},\quad Y_2:=\left\{\psi\in Y\Big| \int_0^1\psi(t)dt=0\right\}\\
\dis X_1:=Y_1,\quad X_2:=X\cap Y_2,\end{array}
$$
and the projections $P : Y\to Y_1$ and $Q=\I-P : Y\to Y_2$ defined as $P\psi:= \int_0^1\psi(t)dt$, we get the next result related to averaging theory \cite{GH,SVM}.
\begin{theorem}\label{thaver}
Suppose \eqref{c3} and the Melnikov-like function
\begin{equation}\label{e16}
M_{aver}(\al):=f(\al,0)+h(H^{-1}(-d))\int_0^1\ga(t)dt
\end{equation}
has a simple zero at some $\al_0\in\R$.  Then equation \eqref{e14} has a unique $1-$periodic solution $\psi(t)=\al_0+O(\ep)$ for any $\ep$ sufficiently small.
\end{theorem}
\section{Forced anti-symmetric periodic solutions}\label{for}
In this section we consider the problem of existence of periodic solutions of the differential equation \eqref{ie1e}. As we apply topological degree arguments \cite{Ch,F2,M}, we introduce a new parameter $\la\in[0,1]$, changing $f(\psi)$ with $\la f(\psi)$, and $\ga(t)$ with $\la\ga(t)$:
\begin{equation}\label{ee1}
\ddot \psi = \la f(\psi) + \la\ga(t)h\(H^{-1}\(\frac{\dot\psi^2}{2}-\la F(\psi)-d\)\).
\end{equation}
Note that, for $\la=1$ we have equation \eqref{e1c} with $\ep=1$ and with $\la=0$ we obtain the equation $\ddot \psi = 0$.

For any $w\in J$ we have:
\[
\frac{d}{dw} H^{-1}(w) = \frac{1}{H'(H^{-1}(w))} = \frac{1}{h(H^{-1}(w))}
\]
hence from equation \eqref{ee1} we obtain:
\begin{equation}\label{e5}
\frac{d}{dt}H^{-1}\(\frac{\dot\psi(t)^2}{2}-\la F(\psi(t))-d\) = \la\ga(t)\dot\psi(t),
\end{equation}
which implies
$$
\begin{array}{l}
\dis H^{-1}\(\frac{\dot\psi(t)^2}{2}-\la F(\psi(t))-d\) \\
\dis = H^{-1}\(\frac{\dot\psi(t_0)^2}{2}-\la F(\psi(t_0))-d\) + \la\int_{t_0}^t\ga(s)\dot\psi(s) ds
\end{array}
$$
that is
\begin{equation}\label{e5a}
\begin{array}{l}
\dis \frac{\dot\psi(t)^2}{2}-\la F(\psi(t))\\ \dis
= d + H\( H^{-1}\(\frac{\dot\psi(t_0)^2}{2}-\la F(\psi(t_0))-d\)\right.\\ \dis
 \left.+ \la\int_{t_0}^t\ga(s)\dot\psi(s) ds \).
\end{array}
\end{equation}

To proceed we recall the following Bihari's inequality \cite{Bi}.

\begin{theorem}\label{bih}
If $w(t)$ is a nonnegative continuous function such that:
\[
w(t) \le \al + \int_0^t f(s)g(w(s)) ds
\]
with a constant $\al>0$, $g:[0,\infty[\to [0,\infty[$ nondecreasing continuous with $g(u)>0$ if $u>0$, and $f(t)$ is continuous with $f(t)\ge 0$ for any $t\ge 0$ then, with $G(x)=\int_1^x \frac{1}{g(u)}du$,
\[
w(t) \le G^{-1}\( G(\al) + \int_0^t f(s)\,ds\)
\]
for all $t\ge 0$ for which $G(\al) + \int_0^t f(s)\,ds$ belongs to the domain of $G^{-1}$.
\end{theorem}

Suppose that \eqref{e5}, or \eqref{e5a}, has a $1-$periodic or a homoclinic, solution $\psi(t)$. Then at some point $t_0$ it must result $\dot\psi(t_0)=0$. Motivated by this, we prove the following.
\begin{lemma}\label{lem2.1}
Let $\Gamma:=\int_0^1|\ga(s)|ds$ and assume the following conditions hold:
\begin{itemize}
\item[i)] $|F(\psi)|\le M$ for any $\psi\in \R$,
\item[ii)] $|H(z)|\le g(|z|)$  for any $z\in \R$,
\item[iii)] $[-d-\la\sup_{\psi\in\R}F(\psi), -d-\la\inf_{\psi\in\R}F(\psi)]\subset J$ for any $\la\in[0,1]$,
\item[iv)] $\dis\Gamma< \int_{N+\Gamma\sqrt{2(M+|d|)}}^\infty\frac{du}{\sqrt{2g(u)}}$,
\end{itemize}
where
\[
N := \sup\{ |H^{-1}(-\la F(\psi)-d)| \mid \psi\in\R,\;\la\in[0,1]\},
\]
and $g:[0,\infty[\to[0,\infty[$ is a nondecreasing continuous function satisfying $g(r)>0$ for $r>0$. Then there exists a positive constant $M_1$ such that for any $\la\in[0,1]$, any $1-$periodic solution $\psi(t)$ of equation \eqref{e5a} whose derivative vanishes at $t_0\in [0,1[$ satisfies:
\[
\|\dot\psi\|_0 \le M_1.
\]
\end{lemma}
\begin{proof}
From equation \eqref{e5a} we get:
\[\begin{array}{l}
\dis\frac{1}{2}\dot \psi(t)^2 \le |d|+M \\
\dis + g\(\left |H^{-1}\(-\la F(\psi(t_0))-d\)+\la\int_{t_0}^t\ga(s)\dot\psi(s) ds\right |\)    \\
\dis \le |d|+M \\
\dis + g\(\left |H^{-1}\(-\la F(\psi(t_0))-d\)\right | + \int_{t_0}^t|\ga(s)\dot\psi(s)| ds\)   \\
\dis \le |d|+M + g\(N + \int_{t_0}^t|\ga(s)\dot\psi(s)| ds\)
\end{array}\]
for $t\geq t_0$. Note that here we need that
$$
[-d-\la\sup_{\psi\in\R}F(\psi),-d-\la\inf_{\psi\in\R}F(\psi)]\subset J
$$
for any $\la\in[0,1]$. Then
\[\begin{array}{l}
\dis |\dot \psi(t)| \le\sqrt{2(|d|+M) + 2g\(N+\int_{t_0}^t|\ga(s)\dot\psi(s)| ds\)} \\
\dis \le \sqrt{2(|d|+M)} + \sqrt{2g\(N + \int_{t_0}^t|\ga(s)\dot\psi(s)| ds\)}.
\end{array}\]
Setting $\dis w(t):=N+\int_{t_0}^t|\ga(s)\dot\psi(s)| ds$:
\[
w'(t) \le |\ga(t)|\sqrt{2(|d|+M)} + |\ga(t)|\sqrt{2g(w(t))}
\]
and then
\[
w(t) \le N + \Gamma\sqrt{2(|d|+M)} + \int_{t_0}^t |\ga(s)| \sqrt{2g(w(s))}\, ds.
\]
Bihari's inequality implies
\[
w(t) \le G^{-1} \( G\(N + \Gamma\sqrt{2(|d|+M)}\) + \int_{t_0}^t |\ga(s)|\, ds\)
\]
where $G^{-1}$ is the inverse function of
\[
G(x):=\int_1^x \frac{du}{\sqrt{2g(u)}}
\]
and for all $t\in [t_0,t_0+1]$ for which
\[
G\(N + \Gamma\sqrt{2(|d|+M)}\) + \int_{t_0}^t |\ga(s)|\, ds
\]
belongs to the domain of $G^{-1}$. Since $g(r)>0$ for all $r>0$ we get:
\[
G:[0,\infty[ \to \left [-\int_0^1 \frac{du}{\sqrt{2g(u)}}, \int_1^\infty\frac{du}{\sqrt{2g(u)}} \right [ .
\]
So $G\(N + \Gamma\sqrt{2(|d|+M)}\) + \int_{t_0}^t |\ga(s)|\, ds$ belongs to the domain of $G^{-1}$ if and only if
\[
G\(N + \Gamma\sqrt{2(|d|+M)}\) + \int_{t_0}^t |\ga(s)|\, ds < \int_1^\infty\frac{du}{\sqrt{2g(u)}},
\]
and this holds if:
\[
\begin{array}{c}
\dis \int_1^{N + \Gamma\sqrt{2(|d|+M)}} \frac{du}{\sqrt{2g(u)}} + \Gamma < \int_1^\infty\frac{du}{\sqrt{2g(u)}} \\ \dis \Leftrightarrow \; \Gamma < \int_{N + \Gamma\sqrt{2(|d|+M)}}^\infty\frac{du}{\sqrt{2g(u)}} .
\end{array}
\]
Then the thesis follows with
\[
\begin{array}{c}
\dis M_1 = \\
\dis \sqrt{2\( |d|+M + g\(G^{-1} \( G\(N + \Gamma\sqrt{2(|d|+M)}\) + \Gamma\)\)\)}.
\end{array}
\]
The proof is finished. \end{proof}

Now we assume the following hold:
\begin{itemize}
\item[$(A_1)$] $f(\psi)$ is odd and continuous,
\item[$(A_2)$] $\ga(t)$ is odd and continuous,
\item[$(A_3)$] $\left[-d-M, -d+M+\frac{M_1^2}{2}\right]\subset J$.
\end{itemize}
Note $(A_3)$ implies iii) of Lemma \ref{lem2.1}. Suppose that $\psi(t)$ solves equation \eqref{ee1}. Then $-\psi(-t)$ is also a solution of \eqref{ee1}. Let $X, Y,Z$ be, resp., the space of $C^2$, $C^0$, and $C^1$, $1-$periodic odd functions. By $(A_3)$, we may take $M_1<\bar M_1$ so that
$$
\left[-d-M, -d+M+\frac{\bar M_1^2}{2}\right]\subset J.
$$
Since $\psi(t)=\int_0^t\dot\psi(s)ds$ implies $\|\psi\|_0\le\|\dot\psi\|_0$, we can take the norm $\|\psi\|_1:=\|\dot\psi\|_0$ on $Z$. Let $B_{\bar M_1}(0)$ be the ball in $Z$ of radius $\bar M_1$ and center $0$. Then the map:
\[
{\cal F}_\la:\psi \mapsto \la f(\psi) + \la\ga(t)h\(H^{-1}\(\frac{\dot\psi^2}{2}-\la F(\psi)-d\)\)
\]
from $B_{\bar M_1}(0)$ to $Y$ is well-defined.  Moreover, for any function $v(t)\in Y$ the equation
\[
\ddot\psi(t) = v(t)
\]
has the unique odd $1-$periodic solution
$$
\begin{array}{l}
\dis\psi(t)=(L_pv)(t):=\int_0^t\int_0^s v(\tau)d\tau ds-t\int_0^1\int_0^s v(\tau)d\tau ds\\
\dis =\int_0^t(t-s)v(s)ds+t\int_0^1sv(s)ds\in X
\end{array}
$$
with $\|\psi\|_2\le 2\|v\|_0$, since $\|\psi\|_0\le\|\dot \psi\|_0\le 2\|v\|_0$ and $\|\psi\|_2=\max\{\|\psi\|_0,\|\dot\psi\|_0,\|\ddot\psi\|_0\}\le 2\|v\|_0$. Hence $L_p : Y\to X$ is continuous and so $L_p{\cal F}_\la : B_{\bar M_1}(0)\to X$ is well-defined and continuous.  Arzela-Ascoli theorem gives the compactness of the embedding $X\subset Z$. Then $L_p{\cal F}_\la : B_{\bar M_1}(0)\to Z$ is well-defined, continuous and compact. Next from Lemma \ref{lem2.1} it follows that equation $\psi =L_p{\cal F}_\la\psi$ has no solutions in the boundary of ball $B_{\bar M_1}(0)$. Hence we get
$$
\deg\(I-L_p{\cal F}_1,B_{\bar M_1}(0),0\)=1.
$$
So, we proved the following result.

\begin{theorem}\label{thm2.2}
Assume i), ii) and iv) of Lemma \ref{lem2.1}, $(A_1)$, $(A_2)$, $(A_3)$ and $h$ is continuous satisfying $(C_1)$. Then equation
\begin{equation}\label{eee1}
\ddot \psi = f(\psi) + \ga(t)h\(H^{-1}\(\frac{\dot\psi^2}{2}-F(\psi)-d\)\)
\end{equation}
has a $1-$periodic odd solution.
\end{theorem}

In a similar way we may also consider the problem of existence of solutions of \eqref{eee1} satisfying the condition
\begin{equation}\label{boun1}
\psi(t+q)=\psi(t) +2\pi
\end{equation}
for some $q\in\N$ by supposing $(C_2)$ in addition. So we are looking for $q-$shift-periodic solutions. This is accomplished by looking for solutions $\psi(t)$ of \eqref{ee1} satisfying $\psi(t)=\vphi(t) + \frac{2\pi}{q}t$ for all $t\in\R$ with $\vphi(t+q)=\vphi(t)$. Plugging the equality into \eqref{ee1} we get the equation for $\vphi(t)$
\begin{equation}\label{ee2}
\begin{array}{l}
\ddot\vphi = \la f\(\vphi+\frac{2\pi}{q}t\) + \\ \dis \la\ga(t)h\(H^{-1}\(\frac{1}{2}\(\dot\vphi+\frac{2\pi}{q}\)^2-\la F\(\vphi+\frac{2\pi}{q}t\)-d\)\).
\end{array}
\end{equation}
As we observed previously, if $f(\psi)$ and $\ga(t)$ are odd and $\vphi(t)$ is a solution of \eqref{ee2} then $-\vphi(-t)$ is also a solution of the same equation \eqref{ee2}. Then we see, as in the proof of Lemma \ref{lem2.1}, that the following inequality holds for any $q-$periodic function $\vphi(t)$ such that $\dot\vphi(t_0)=0$:
$$
\begin{array}{l}
\dis\frac{1}{2}\(\dot\vphi(t)+\frac{2\pi}{q}\)^2\le M + |d| \\
\dis + H\Bigg( H^{-1}\(\frac{2\pi^2}{q^2}-\la F\(\vphi(t_0)+\frac{2\pi t_0}{q}\)-d\) \\ \dis + \la\int_{t_0}^t\ga(s)\left [\dot\vphi(s) +\frac{2\pi}{q}\right ]ds \Bigg)     \\
\dis \le M + |d| + g\Bigg( \left |H^{-1}\(\frac{2\pi^2}{q^2} -\la F\(\vphi(t_0)+\frac{2\pi t_0}{q}\)-d\)\right | \\ \dis + \int_{t_0}^t|\ga(s)|\left |\dot\vphi(s) +\frac{2\pi}{q}\right |ds \Bigg) \\
\dis \le M + |d| + g\( N_q + \int_{t_0}^t|\ga(s)|\left |\dot\vphi(s) +\frac{2\pi}{q}\right |ds \)
\end{array}
$$
where
\begin{equation}\label{ee3}
N_q=\sup\left \{ |H^{-1}(x)| \;\Bigg|\; \frac{2\pi^2}{q^2} - M\le x+d \le \frac{2\pi^2}{q^2} + M\right \}.
\end{equation}
Note that we need that $\frac{2\pi^2}{q^2}-\la F\(\vphi(t_0)+\frac{2\pi t_0}{q}\)-d\in J$ (the domain of $H^{-1}$) in order that the inequality makes sense.  Next arguing as before with $$
w(t)=N_q+\int_{t_0}^t |\ga(s)|\left |\dot\vphi(s) +\frac{2\pi}{q}\right |ds,
$$
we get:
\[
\begin{array}{l}
\dis \dot w(t) \le |\ga(t)|\sqrt{2[M + |d| + g(w(t))]} \\ \dis \le |\ga(t)|\sqrt{2(M + |d|)} + |\ga(t)|\sqrt{2 g(w(t))}
\end{array}
\]
and then
\[
w(t) \le N_q + \sqrt{2(M + |d|)}\bar\Gamma + \int_{t_0}^t  |\ga(s)|\sqrt{2 g(w(s))}ds
\]
for any $t\in[t_0,t_0+q]$, with $\bar\Gamma:=\int_0^q|\ga(s)|ds$. As before Bihari's inequality implies:
\[
\begin{array}{l}
\dis w(t) \le G^{-1} \( G\(N_q + \bar\Gamma\sqrt{2(|d|+M)}\) + \int_{t_0}^t |\ga(s)|\, ds\),\\ \dis \; \forall t\in[t_0,t_0+q]
\end{array}
\]
and then, since $\vphi(t)$ is $q-$periodic:
\[
\begin{array}{l}
\dis |\dot\vphi(t)|\le M_2:=\frac{2\pi}{q} +\Bigg(2\Big(M+|d|\\ \dis +g\Big(G^{-1}\Big( G\(N_q + \bar\Gamma\sqrt{2(|d|+M)}\)+ \bar\Gamma\Big) \Big)\Big)\Bigg)^{1/2},
\end{array}
\]
for all $t\in\R$. Thus, arguing as before we conclude with the following.
\begin{theorem}\label{thm2.3}
Assume i), ii)  of Lemma \ref{lem2.1}, $(C_2)$, $(A_1)$, $(A_2)$, $h$ is continuous satisfying $(C_1)$, and
\begin{itemize}
\item[iv')] $\dis\bar\Gamma< \int_{N_q+\bar\Gamma\sqrt{2(M+|d|)}}^\infty\frac{du}{\sqrt{2g(u)}}$,
\item[$(A_3')$] $\left[\frac{2\pi^2}{q^2}-d-M, \frac{2\pi^2}{q^2}-d+M+\frac{M_2^2}{2}\right]\subset J$.
\end{itemize}
Then equation \eqref{eee1} has a $q-$shift-periodic solution, i.e., satisfying \eqref{boun1}.
\end{theorem}
\begin{remark}\label{rr4} Since $\frac{2\pi^2}{q^2}>0$ it is enough that $J\supset [a,\infty[$ for some $a\in\R$ and $-d>M+a$ for getting $(A_3)$ and $(A_3')$.
\end{remark}

We conclude this section with the following.
\begin{theorem}\label{thm2.4}
Assume i), ii) and iv) of Lemma \ref{lem2.1}, $(A_1)$, $(A_3)$ and $h$ is continuous satisfying $(C_1)$. Moreover, suppose it holds
\begin{itemize}
\item[$(A_4)$] $\ga(t+\frac{1}{2})=-\ga(t)$ for any $t\in [0,1]$ and $\ga$ is continuous.
\end{itemize}
Then equation \eqref{eee1} has a $1-$periodic solution satisfying
\begin{equation}\label{anti}
\psi\(t+\frac{1}{2}\)=-\psi(t)\quad\forall t\in [0,1].
\end{equation}
\end{theorem}
\begin{proof}
Let $X, Y,Z$ be, resp., the space of $C^2$, $C^0$, and $C^1$, $1-$periodic functions satisfying \eqref{anti}. By $(A_3)$, we may take $M_1<\bar M_1$ so that
$$
\left[-d-M, -d+M+\frac{\bar M_1^2}{2}\right]\subset J.
$$
Furthermore, for any function $v\in Y$ the equation
\[
\dot\psi(t) = v(t)
\]
has the unique $1-$periodic solution
$$
\psi(t)=(L_av)(t):=\int_0^tv(s)ds-\frac{1}{2}\int_0^{1/2}v(s)ds
$$
satisfying \eqref{anti}. Clearly $\|\psi\|_0\le\frac{5}{4}\|v\|_0$, so $L_a : Y\to Z$ is continuous. This gives also $\|\psi\|_0\le\frac{5}{4}\|\dot\psi\|_0$. So we can take the norm $\|\psi\|_1:=\|\dot\psi\|_0$ on $Z$. Like above, let $B_{\bar M_1}(0)$ be the ball in $Z$ of radius $\bar M_1$ and center $0$. Then $L_a^2{\cal F}_\la : B_{\bar M_1}(0)\to Z$ is well-defined, continuous and compact. Next from Lemma \ref{lem2.1} it follows that equation $\psi =L_a^2{\cal F}_\la\psi$ has no solutions in the boundary of ball $B_{\bar M_1}(0)$. Hence we get
$$
\deg\(I-L_a^2{\cal F}_1,B_{\bar M_1}(0),0\)=1.
$$
The proof is finished.
\end{proof}

\begin{remark}
Condition \eqref{anti} means that $\psi(t)$ is $\frac{1}{2}-$anti-periodic \cite{AF,H} and of course such $\psi$ is $1-$periodic.
\end{remark}
\section{Applications to PT-symmetric dimer with varying gain-loss}\label{appl}

First, we apply results of Section \ref{bif}. Following the Introduction, we study the coupled system
\begin{equation}\label{e1e}
\begin{array}{l}
\dis \ddot \psi=4\(c^2\sin\psi+\ep\ga(t)z\)\\
\dis \dot z=\ep\ga(t)\dot\psi
\end{array}
\end{equation}
where $\ga(t)$ is $1-$periodic and $C^2$-smooth, $\ep$ is small and $c>0$. It has a form of \eqref{e1} with $f(\psi)=4c^2\sin\psi$ and $h(z)=4z$. So now $F(\psi)=4c^2(1-\cos\psi)$ and $H(z)=2z^2$ with either $I=]-\infty,0[$ or $I=]0,\infty[$, $J=]0,\infty[$, and $(C_1)$ holds. Then \eqref{e1b} has the form
\begin{equation}\label{e1be}
z=\pm\sqrt{\frac{\dot\psi^2}{4}+2c^2(\cos\psi-1)-\frac{d}{2}},
\end{equation}
and \eqref{e1c} is now
\begin{equation}\label{e1ce}
\ddot \psi=4\(c^2\sin\psi\pm\ep\ga(t)\sqrt{\frac{\dot\psi^2}{4}+2c^2(\cos\psi-1)-\frac{d}{2}}\).
\end{equation}
The unperturbed equation of \eqref{e1ce}  is given by (see \eqref{e2})
\begin{equation}\label{e2e}
\ddot \psi=4c^2\sin\psi,
\end{equation}
which possesses a family of periodic solutions \cite[pp. 114-115]{L}
$$
\psi_{k,c}(t)=\pi+2\arcsin\(k\sn\(2ct,k\)\)
$$
with periods
$$
T(k,c)=\frac{2}{c}K(k),
$$
where $K(k)$ is the complete elliptic integral of the first kind, $\sn$ is a Jacobi elliptic function and $k\in(0,1)$. Note $T(k,c)\to\infty$ as $k\to1$ and $T(k,c)\to\frac{\pi}{c}$ as $k\to0$. Moreover it holds
$$
\dot\psi_{k,c}(t)=4ck\cn\(2ct,k\),
$$
where $\cn$ is a Jacobi elliptic function. Finally, we note
$$
\frac{\dot\psi_{k,c}^2}{2}+4c^2\(\cos\psi_{k,c}-1\)=-8c^2(1-k^2).
$$
So we suppose
\begin{equation}\label{c1e}
d<-8c^2(1-k^2)
\end{equation}
to get $(H_2)$. Resonance assumption $(R)$ has the form
\begin{itemize}
\item[$(R_p)$] \ \ \ $K(k)=\frac{cp}{2q}$ for $p,q\in\N$.
\end{itemize}
Applying Theorem \ref{th1}, we obtain the following result.

\begin{theorem}\label{th1e}
Suppose $(R_p)$, \eqref{c1e} and that there is an $\al_0\in\R$ such that
$$
\begin{array}{c}
\dis \int_0^T\ga(t+\al_0)\cn\(2ct,k\)dt=0\\ \dis \quad\text{and}\quad \int_0^T\dot\ga(t+\al_0)\cn\(2ct,k\)dt\ne0.
\end{array}
$$
Then there is a $\de>0$ such that for any $0\ne \ep\in (-\de,\de)$ there are precisely two $p-$periodic solutions $(\psi_\pm,z_\pm)$ of \eqref{e1e} with
$$
\begin{array}{c}
\dis (\psi_\pm(t),z_\pm(t))\\ \dis =\(\psi_{k,c}(t-\al_0),\pm\sqrt{-4c^2(1-k^2)-\frac{d}{2}}\)+O(\ep).
\end{array}
$$
\end{theorem}

Next, $(C_2)$ clearly holds. $(H_3)$ is satisfied for \cite[p. 117]{L}
$$
\psi_{k,c}^l(t)=\pi+2\am\(\frac{2ct}{k},k\)
$$
with periods
$$
T^l(k,c)=\frac{K(k)k}{c},
$$
where $\am$ is a Jacobi elliptic function and $k\in(0,1)$. Moreover it holds
$$
\dot\psi_{k,c}^l(t)=\frac{4c}{k}\dn\(\frac{2ct}{k},k\),
$$
where $\dn$ is a Jacobi elliptic function. Finally, we note
$$
\frac{\dot\psi_{k,c}^{l2}}{2}+4c^2\(\cos\psi_{k,c}^l-1\)=\frac{8c^2}{k^2}(1-k^2).
$$
So we suppose
\begin{equation}\label{c1ee}
d<\frac{8c^2}{k^2}(1-k^2)
\end{equation}
to get $(H_2')$. Resonance assumption ($R'$) has the form
\begin{itemize}
\item[$(R_p')$] \ \ \ $K(k)k=\frac{cp}{q}$ for $p,q\in\N$.
\end{itemize}
Applying Theorem \ref{th1b}, we obtain the following result.

\begin{theorem}\label{th1ee}
Suppose $(R_p')$, \eqref{c1ee} and that there is an $\al_0\in\R$ such that
$$
\begin{array}{c}
\dis \int_0^T\ga(t+\al_0)\dn\(\frac{2ct}{k},k\)dt=0\\ \dis \quad\text{and}\quad \int_0^T\dot\ga(t+\al_0)\dn\(\frac{2ct}{k},k\)dt\ne0.
\end{array}
$$
Then there is a $\de>0$ such that for any $0\ne \ep\in (-\de,\de)$ there are precisely two $p-$shift-periodic solutions $(\psi_\pm^l,z_\pm^l)$ of \eqref{e1e} with
$$
\begin{array}{c}
\dis (\psi_\pm^l(t),z_\pm^l(t))\\ \dis =\(\psi_{k,c}(t-\al_0),\pm\sqrt{\frac{4c^2}{k^2}(1-k^2)-\frac{d}{2}}\)+O(\ep).
\end{array}
$$
\end{theorem}

Similarly we can study heteroclinic Melnikov bifurcation to \eqref{e1e}. The unperturbed equation \eqref{e2e} has a heteroclinic cycle
$$
\psi_{c,\pm}(t)=\pi\pm2\arctan\(\sinh(2ct)\)
$$
with
$$
\dot\psi_{c,\pm}(t)=\pm4c\sech(2ct),
$$
and
$$
\frac{\dot\psi_{c,\pm}^2}{2}+4c^2\(\cos\psi_{c,\pm}-1\)=0.
$$
So we consider
\begin{equation}\label{c2ee}
d<0,
\end{equation}
which is a limit of \eqref{c1e} and \eqref{c1ee} as $k\to 1^-$. Applying Theorem \ref{th2}, we have the following result.

\begin{theorem}\label{th2e}
Suppose \eqref{c2ee} and that there is an $\al_0\in\R$ such that
$$
\begin{array}{c}
\dis \int_{-\infty}^\infty\ga(t+\al_0)\sech(2ct)dt=0\\ \dis \quad\text{and}\quad \int_{-\infty}^\infty\dot\ga(t+\al_0)\sech(2ct)dt\ne0.
\end{array}
$$
Then there is a $\de>0$ such that for any $0\ne \ep\in (-\de,\de)$ there are precisely four bounded solutions $(\psi_\pm^\pm,z_\pm^\pm)$ of \eqref{e1e} with
$$
\begin{array}{c}
\dis (\psi_\pm^\pm(t),z_\pm^\pm(t))\\ \dis =\(\pi\pm2\arctan\(\sinh(2c(t-\al_0))\),\pm\sqrt{\frac{-d}{2}}\)+O(\ep).
\end{array}
$$
\end{theorem}

Clearly, Remarks \ref{rr1} and \ref{rr2} can be applied to \eqref{e1e}. To apply Remark \ref{rr3}, we shift \eqref{e1ce} by $\psi(t)=\pi+\vphi(t)$ to get
\begin{equation}\label{e1ee}
\begin{array}{l}
\dis \dot\vphi=\th,\\
\dis \dot\th=4\(-c^2\sin\vphi\pm\ep\ga(t)\sqrt{\frac{\th^2}{4}-2c^2(\cos\vphi+1)-\frac{d}{2}}\).
\end{array}
\end{equation}
The unperturbed equation of \eqref{e1ee} when $\ep=0$ is just the pendulum equation
$$
\begin{array}{l}
\dis \dot\vphi=\th,\\
\dis \dot\th=-4c^2\sin\vphi,
\end{array}
$$
which possesses well-known families of solutions \cite[pp. 115-117]{L} mentioned above, and so Remark \ref{rr3} can be used to \eqref{e1ee}.

Now we suppose in \eqref{e1e} that $c\to0$ as $\ep\to0$, so we consider
\begin{equation}\label{e13e}
\begin{array}{l}
\dis \ddot \psi=4\(\ep \wt c(\ep)\sin\psi+\ep\ga(t)z\)\\
\dis \dot z=\ep\ga(t)\dot\psi
\end{array}
\end{equation}
where $\ep>0$ is small and $\wt c(\ep)>0$ is $C^2$-smooth around $\ep=0$. Then \eqref{e13e} has a form of \eqref{e13}. Assumption \eqref{c3} reads now
\begin{equation}\label{c3e}
d<0.
\end{equation}
The Melnikov function \eqref{e16} has now the form
\begin{equation}\label{e16e}
M_{aver}(\al)=\wt c(0)\sin \al\pm\sqrt{\frac{-d}{2}}\int_0^1\ga(t)dt.
\end{equation}
By applying Theorem \ref{thaver} we get the next.

\begin{theorem}\label{thaver2}
Assume that
\begin{equation}\label{c4e}
\wt c(0)>0,\quad \sqrt{\frac{-d}{2}}\left|\int_0^1\ga(t)dt\right|<\wt c(0),
\end{equation}
then there are exactly four $1-$periodic solutions of \eqref{e1e} with $c^2=\ep \wt c(\ep)$ and $\ep>0$ small.
\end{theorem}
\begin{proof}
Condition \eqref{c4e} gives that \eqref{e16e} has 2 simple roots over $[0,2\pi)$. The proof is finished.
\end{proof}

On the other hand, if
$$
\sqrt{\frac{-d}{2}}\left|\int_0^1\ga(t)dt\right|>\wt c(0),
$$
such \eqref{e1e} has no $1-$periodic solutions for $\ep>0$ small.

In the last part of this section, we apply results of Section \ref{for} to \eqref{ee1ee} when $\ga(t)$ is $1-$periodic and continuous, and $c>0$. First, assumptions i), ii) of Lemma \ref{lem2.1} hold with $M=8c^2$ and $g(z)=2z^2$. Since $\int^\infty \frac{dz}{2z}=\infty$, assumptions iii) and iv) of Lemma \ref{lem2.1} are satisfied provided
\begin{equation}\label{ccc}
d<-8c^2.
\end{equation}
We already verified $(C_1)$,  $(C_2)$, and $(A_1)$ holds as well. Since now $J=]0,\infty[$, it follows from Remark \ref{rr4} that ($A_3$) and ($A'_3$) are satisfied when \eqref{ccc} holds.

Summarizing, we see that results of Section \ref{for} can be used to \eqref{ee1ee} under assumption \eqref{ccc} and $\ga(t)$ satisfies either $(A_2)$ or $(A_4)$, respectively. On the other hand, we need
\begin{equation}\label{ccc2}
\frac{\dot\psi^2}{4}\ge2c^2(1-\cos\psi)+\frac{d}{2}.
\end{equation}
Since
$$
\frac{\dot\psi^2}{4}\ge2c^2(1-\cos\psi)+\frac{d}{2}\ge \frac{d}{2},
$$
we see that if
$$
d>0
$$
then $\dot\psi(t)\ne0$, so $\psi(t)$ can be neither periodic nor homoclinic nor heteroclinic. If $d<0$, then by \eqref{c1e} and \eqref{c1ee}, \eqref{ee1ee} may have oscillating solutions for suitable $\ga(t)$. If $d=0$, then condition $\psi(t_0)=0$ and \eqref{ccc2} give $\cos\psi(t_0)=1$, and this situation occurs when $\ga(t)=0$, for instance.

\section{Numerical results}\label{numer}

\begin{figure}[tb!]
\begin{center}
{\includegraphics[width=9cm]{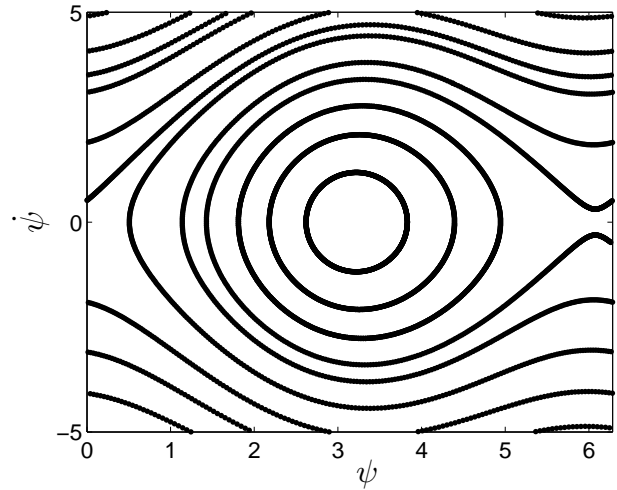}}
\end{center}
\caption{Phase portrait of equation \eqref{ee1ee} with $\gamma=0.01$. }
\label{fig1}
\end{figure}

\begin{figure}[tbhp!]
\begin{center}
\subfigure[$\epsilon=0.2$]{\includegraphics[width=9cm]{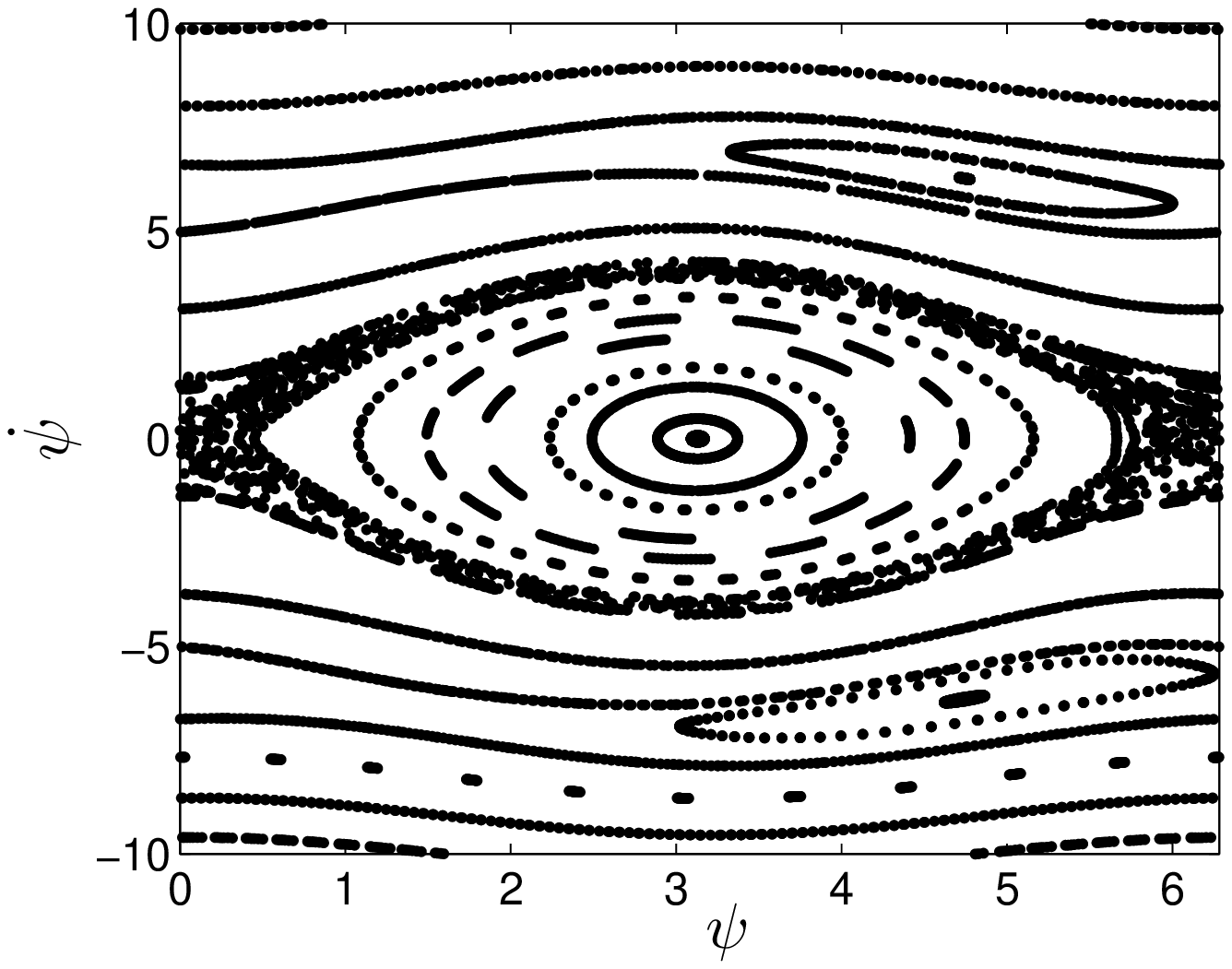}}\\
\subfigure[$\epsilon=2.0$]{\includegraphics[width=9cm]{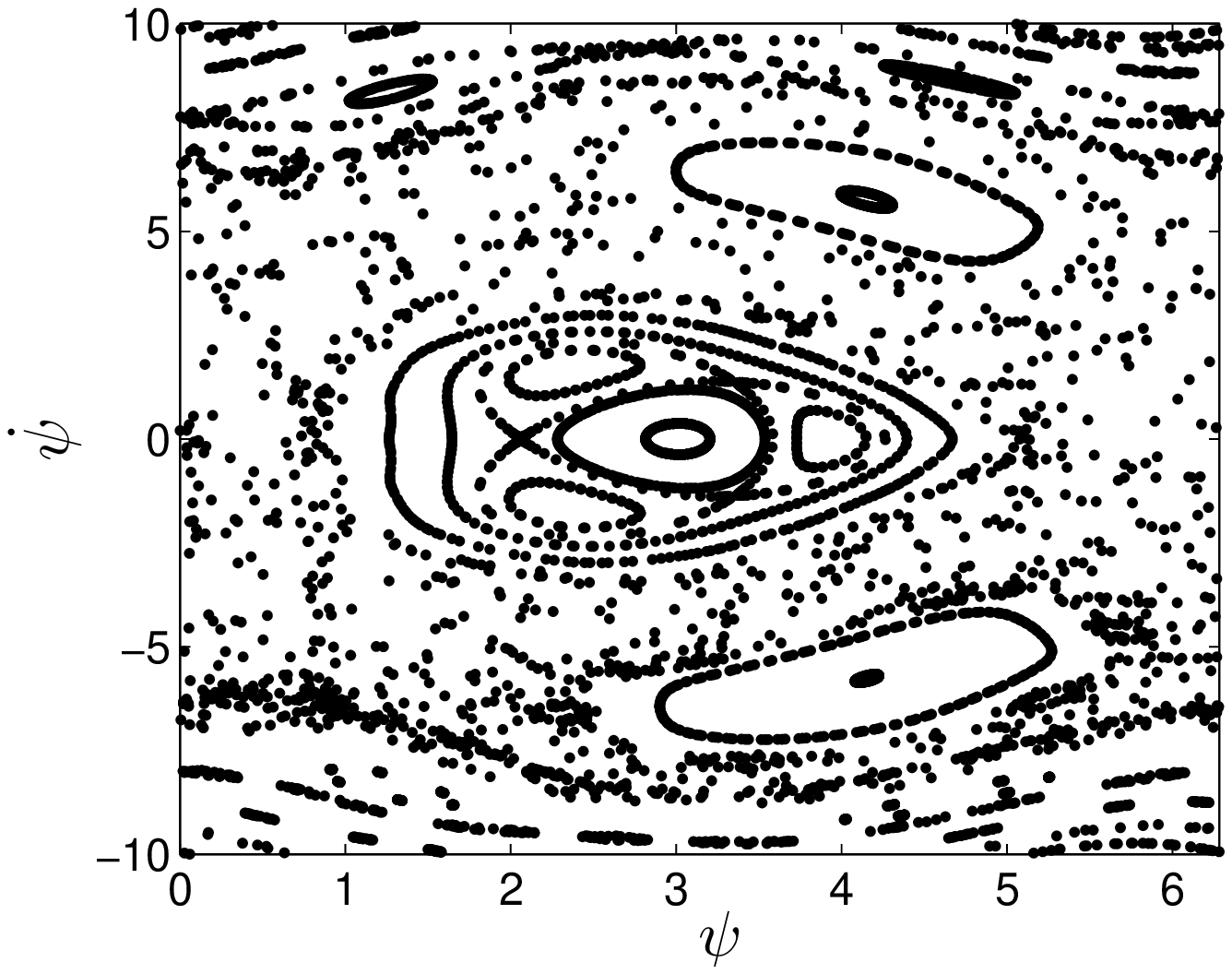}}
\end{center}
\caption{Stroboscopic plots of the system \eqref{ee1ee} with \eqref{gt1} for two values of $\epsilon$.}
\label{fig2}
\end{figure}

\begin{figure}[tbhp!]
\begin{center}
\subfigure[$\epsilon=0.2$]{\includegraphics[width=9cm]{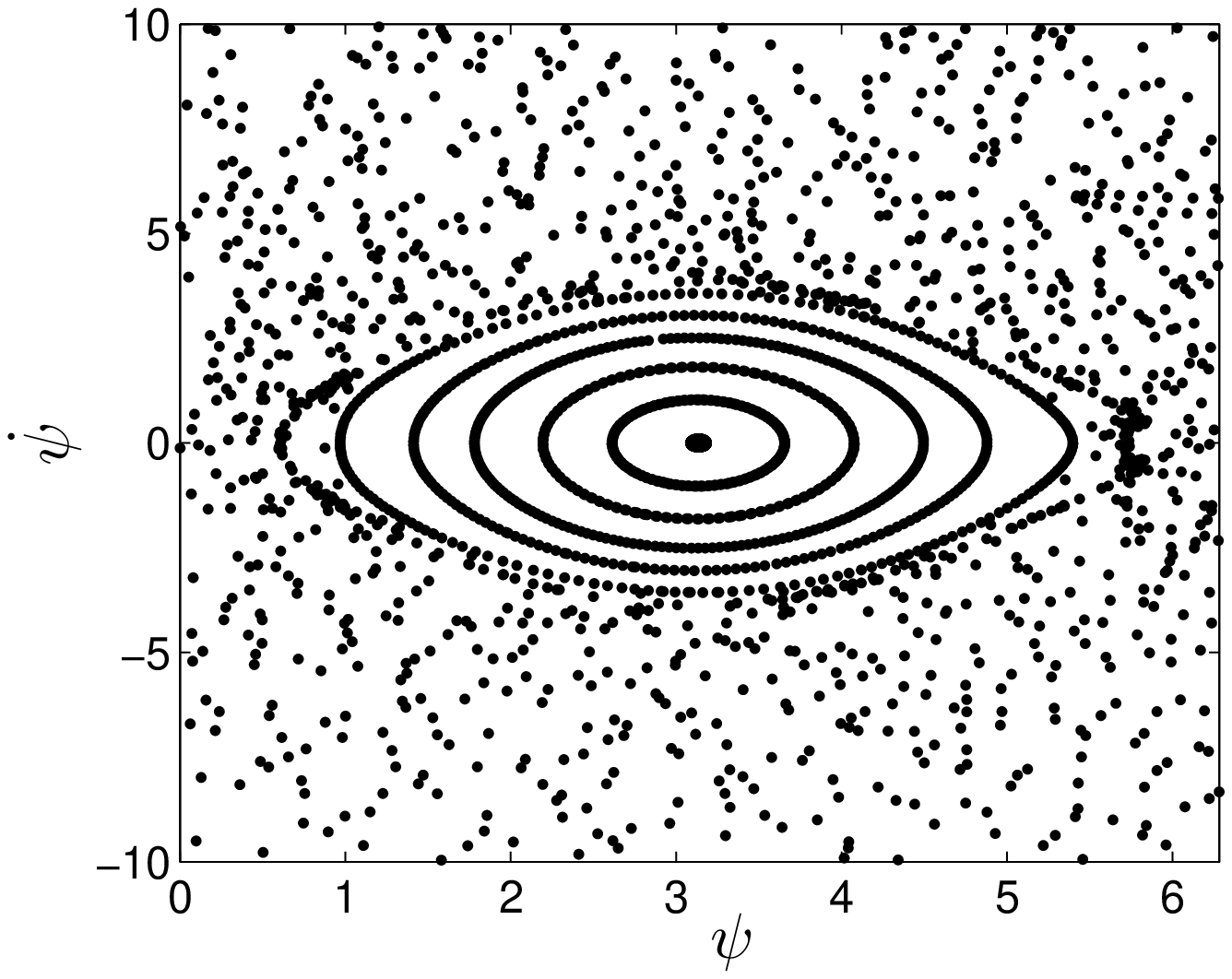}}\\
\subfigure[$\epsilon=0.5$]{\includegraphics[width=9cm]{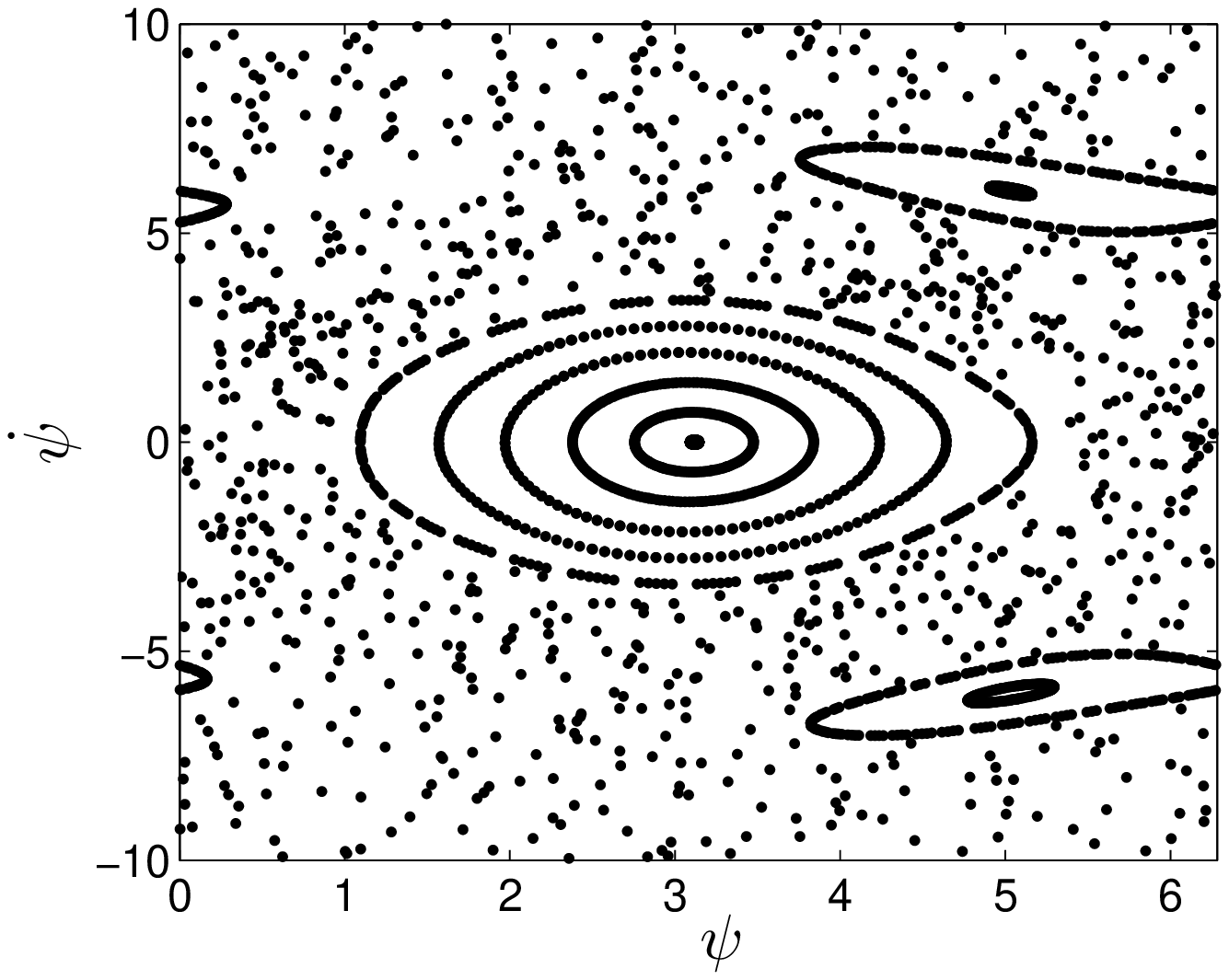}}
\end{center}
\caption{The same as Fig.\ \ref{fig2}, but for $\gamma(t)$ defined in \eqref{gt2}.}
\label{fig3}
\end{figure}

\begin{figure}[tbhp!]
\begin{center}
\subfigure[$c=0.05$]{\includegraphics[width=9cm]{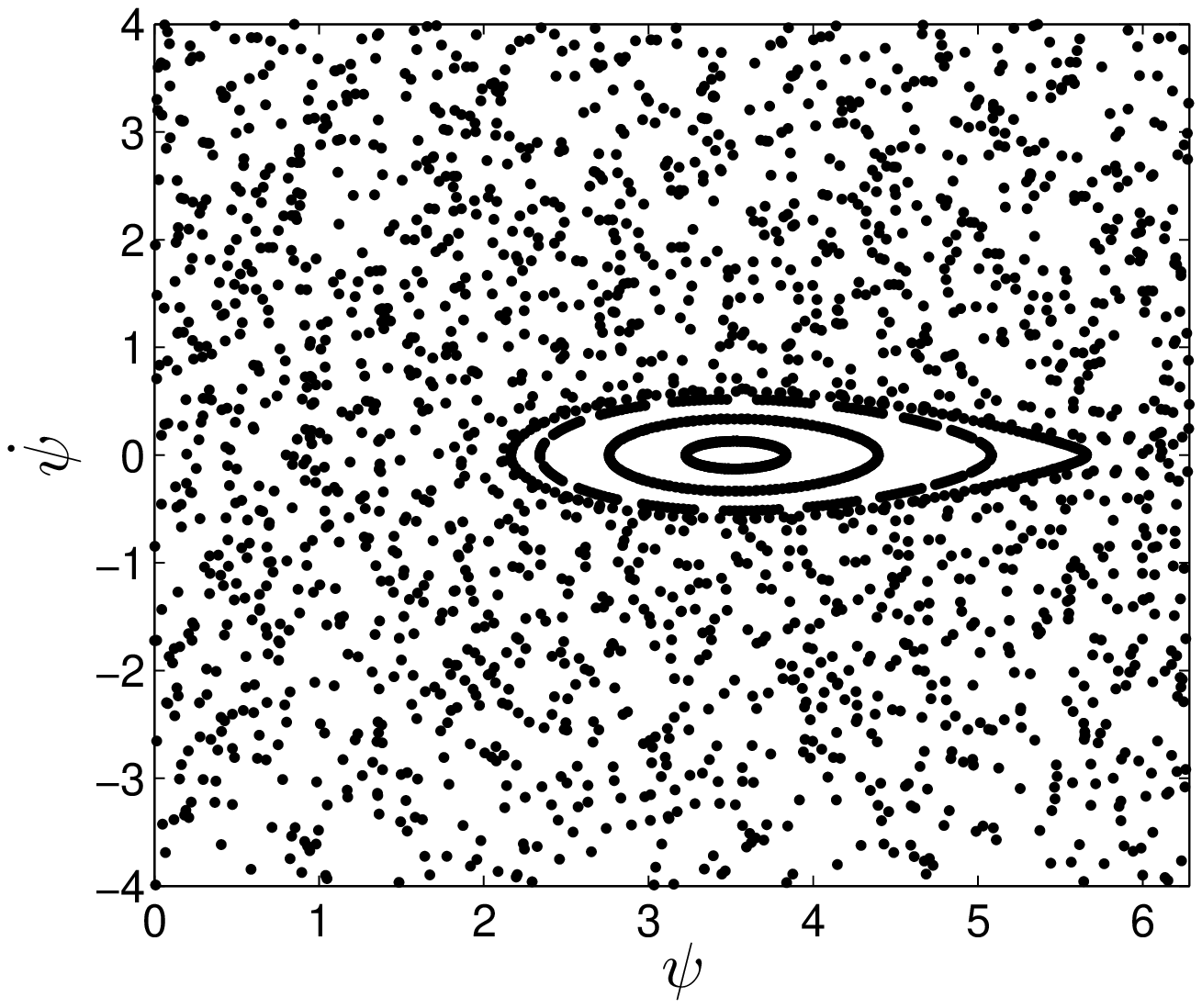}}\\
\subfigure[$c=0.02$]{\includegraphics[width=9cm]{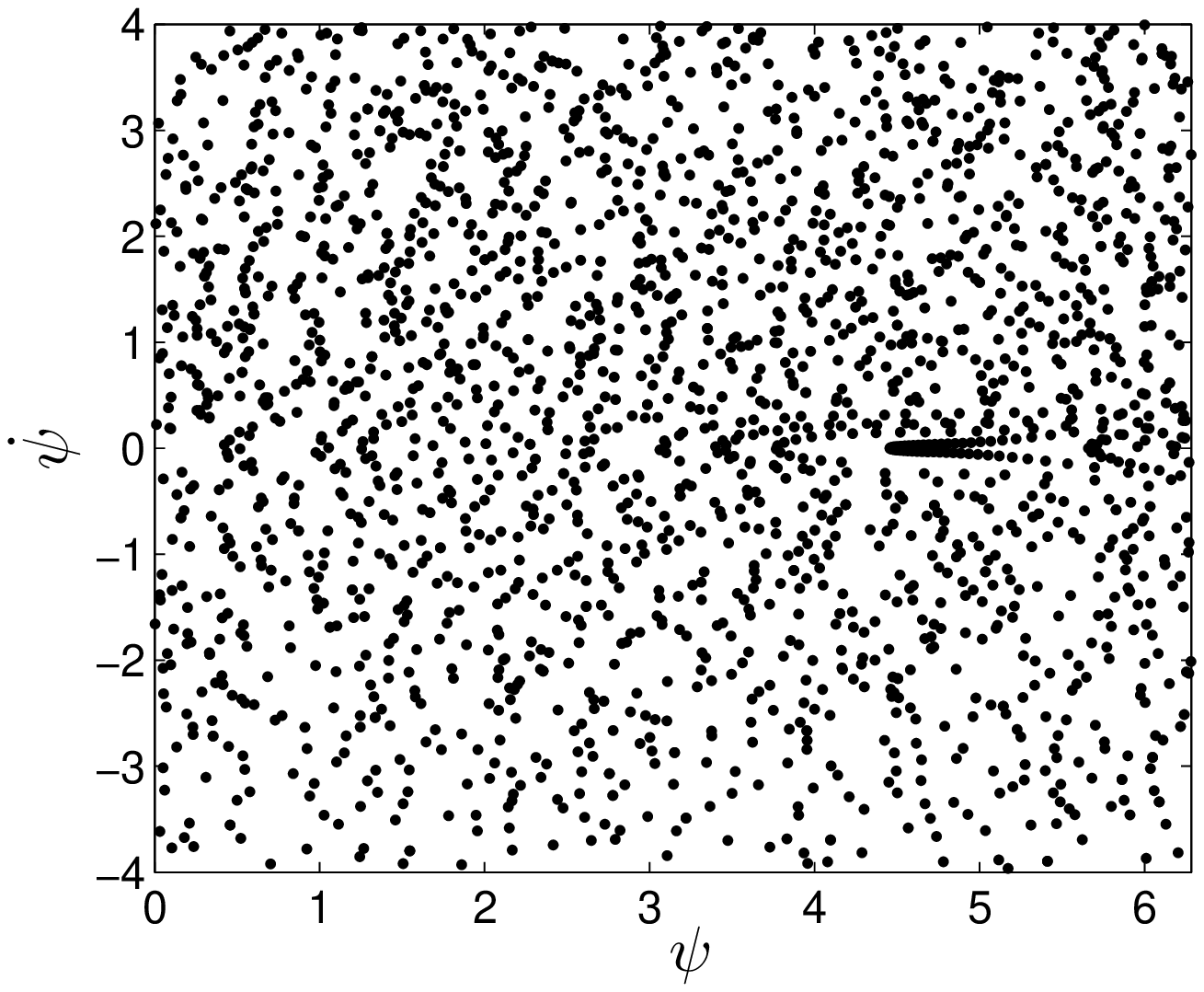}}
\end{center}
\caption{The same as Fig.\ \ref{fig3}, but for (a) $c=0.05$ and (b) $c=0.02$.}
\label{fig4}
\end{figure}

To illustrate the analytical results reported in the previous sections, we integrate the governing equation \eqref{ee1ee} numerically using a multistep Adam-Bashforth-Moulton method. The initial conditions $\psi(0)$ and $\dot{\psi}(0)$ are chosen independently, while $z(0)$ satisfies \eqref{e1be} for a fixed $d\in\mathbb{R}$, such that the solution curves of \eqref{ee1ee} lie on the three-dimensional surface \eqref{tamb3}. The solution trajectories can therefore be conveniently plotted on the, e.g., $(\psi,\dot{\psi})-$plane (cf.\ \eqref{e1c}, \eqref{e1ce}). Throughout the section, we set $d=-8.5$.

First, we consider $c=1$. In the absence of gain-loss parameter, i.e.\ $\gamma=0$, the governing equation reduces to the well-known pendulum equation that has two families of solutions, namely the periodic (libration) and the $p-$shift periodic (rotation) solutions. The solution families are separated by heteroclinic manifolds. In Fig.\ \ref{fig1} the phase portrait of \eqref{ee1ee} is plotted for a small value of constant $\gamma$. In the presence of a non-vanishing gain-loss parameter, while periodic solutions persist, there are no longer $p-$shift periodic solutions. Instead of heteroclinic trajectories, one has homoclinic connections with unbounded solutions outside the invariant manifolds. By increasing the gain-loss parameter $\gamma$, there will be a critical value $\gamma_c$ above which no periodic solutions exist and all solutions become unbounded. This region is referred to as the PT-broken phase regime and the transition has been completely discussed in \cite{P}. For the parameters above, $\gamma_c\approx0.79$.

When $\gamma(t)$ is a periodic function of time, rather than showing the phase portraits with a continuous time, it becomes convenient to represent the solution trajectories in Poincar\'e maps (stroboscopic plots at every period $T=1$). In recurrence maps, a stable periodic orbit will correspond to an elliptical fixed point encircled by closed regions (islands).

As a particular choice, we consider
\begin{equation}
\gamma(t)=\epsilon\cos(2\pi t).
\label{gt1}
\end{equation}
The cosine function is taken to respect the PT-re\-vers\-ibil\-i\-ty of the governing equation. Plotted in Fig.\ \ref{fig2} is the Poincar\'e map of \eqref{ee1ee} with \eqref{gt1}, obtained from various sets of initial conditions using direct numerical integrations of the governing equations. If Fig.\ \ref{fig1} is an ordered dynamical picture, Fig.\ \ref{fig2} have visible chaotic (stochastic) regions.

Fig. \ref{fig2}(a) corresponds to a relatively small value of $\epsilon$. The chaotic layer in the figure is caused by the presence of heteroclinic manifolds connecting $(\psi,\dot{\psi})=(0,0)$ and $(2\pi,0)$, i.e.\ separatrix chaos, as proven in Theorem \ref{th2e}  (see also Remark \ref{rr1}).

Within the region bounded by the chaotic layer, the presence of closed islands is observable. The centers of the islands correspond to stable $p-$periodic solutions shown to persist in Theorem \ref{th1e}. Beyond the chaotic layer, there are also islands with centers corresponding to stable $p-$shift periodic solutions, i.e.\ bounded running modes with period $p$, discussed in Theorem \ref{th1ee}, which do not exist in the system with constant $\ga$.

We have increased the strength of the gain-loss coefficient $\epsilon$. Depicted in Fig.\ \ref{fig2}(b) is the Poincar\'e map for $\epsilon=2$. Despite the significant expansion of the chaotic region, one can observe similar features as in Fig.\ \ref{fig2}(a). The presence of periodic solutions and $p-$shift periodic solutions are guaranteed respectively by Theorems \ref{thm2.2} (and \ref{thm2.4}) and \ref{thm2.3} as the gain-loss function \eqref{gt1} satisfies the condition $(A2)$ and $(A4)$ under the translation $t\to (t-1/4)$. Note as well that the gain-loss strength $\epsilon$ is beyond the critical value of PT-broken phase for constant $\ga$. Hence, similarly to the result reported in \cite{horn13}, we prove and observe a controllable expansion of the exact PT-symmetry region.

On the interesting persistence of $p-$shift periodic solutions observed in Fig.\ \ref{fig2}, it may be unsurprising as the parametric drive \eqref{gt1} has zero average. Nevertheless, the sufficient condition in Theorem \ref{th1ee} does not require a zero-averaged $\ga$. To illustrate it, next we consider the gain-loss function
\begin{equation}
\ga(t)=10^{-2}+\epsilon\cos(2\pi t).
\label{gt2}
\end{equation}
When $\epsilon=0$, one will obtain the phase-portrait in Fig.\ \ref{fig1}.

Shown in Fig.\ \ref{fig3} are the phase-portraits of the system with the gain-loss function \eqref{gt2} for two values of $\epsilon$. In both cases, periodic solutions exist. For $\epsilon=0.2$ in panel (a), we have a similar plot as in Fig.\ \ref{fig1}, where there is no bounded $p-$shift periodic solution. Interestingly, when $\epsilon$ is large enough such that the conditions in Theorem \ref{th1ee} can be satisfied, we observe islands in the region of unbounded solutions, see panel (b). Hence, we obtain (stable) bounded rotating solutions as expected from the theorem.

Finally, we consider the case when both the gain-loss strength and the parameter $c$ are small. We present in Fig.\ \ref{fig4} the Poincar\'e maps of the system for two small values of $c$ and the gain-loss function \eqref{gt2} with $\epsilon=0.2$.

Note that using Theorem \ref{thaver2} (cf.\ Theorem \ref{thaver}), periodic solutions exist if and only if the parameter $c$ exceeds the threshold value
\begin{equation}
\sqrt{\frac{-d}2}\left|\int_0^1\gamma(t)\,dt\right|\approx0.0206.
\label{ms}
\end{equation}
In panel (a) of Fig.\ \ref{fig4}, $c$ is above the value and the presence of periodic solutions can be clearly seen. In panel (b), we decrease the value of $c$ below \eqref{ms} and in good agreement with the theorem there is no more periodic solution. All solutions become unbounded. Therefore, Theorem \ref{thaver2} provides an approximate threshold value for PT-broken phase.

\section*{Acknowledgement}
FB is partially supported by PRIN-MURST {\it Equazioni Differenziali Ordinarie e Applicazioni}. JD is partially supported by Grant GA\v CR P201/11/0768. MF is partially supported by Grant VEGA-MS 1/0071/14. MP is supported by project No.\ CZ.1.07/2.3.00/30.0005 funded by European Regional Development Fund.

\end{document}